\documentclass[10pt]{article}
\usepackage{hyperref}
\usepackage[T1]{fontenc}
\usepackage[ansinew]{inputenc}
\usepackage[english]{babel}
\usepackage{csquotes}
\usepackage{amsmath}
\usepackage[backend=bibtex,url=false,doi=false,style=numeric-comp,firstinits=true]{biblatex}
\renewbibmacro{in:}{}
\bibliography{~/library/jabref_biblio.bib}
\usepackage{amsfonts}
\usepackage{amssymb}
\usepackage{amsthm}
\usepackage{graphicx}
\usepackage{color}
\usepackage{booktabs}
\usepackage{float}
\usepackage{algorithm}

\newtheorem{theorem}{Theorem}
\newtheorem{assumption}{Assumption}
\newtheorem{lemma}{Lemma}
\newtheorem{definition}{Definition}
\newtheorem{corollary}{Corollary}
\newtheorem{prop}{Proposition}

\newcommand{\diag} {\mbox{\rm diag}}
\newcommand{\rank} {\mbox{\rm rank}}
\newcommand{\Var} {\mbox{\rm Var}\,}
\newcommand{\norm}[1]{\left\| #1 \right\|}
\newcommand{\E}{\mathbb{E}}
\newcommand{\T}{\mathbf{T}}
\newcommand{\kron}{\otimes}
\newcommand{\N}{\ensuremath{\mathcal{N}}}
\newcommand{\matlab}{\textsc{matlab}}
\newcommand{\vect}{\mathrm{vec}}
\newcommand{\sign}{\mathrm{sign}}
\renewcommand{\matrix}[1]{\begin{bmatrix} #1 \end{bmatrix}}

\title{\LARGE \bf A kernel-based approach to\\ overparameterized Hammerstein system identification}
\author{Riccardo S. Risuleo, Giulio Bottegal and H\r akan Hjalmarsson
\thanks{R. S. Risuleo, G. Bottegal and H. Hjalmarsson are with the ACCESS Linnaeus Center, School of Electrical Engineering, KTH Royal Institute of Technology, Sweden ({\tt \small risuleo;bottegal;hjalmars@kth.se}).
This work was supported by the European Research Council under the
advanced grant LEARN, contract 267381 and by the Swedish Research
Council under contract 621--2009--4017.
}}
\begin{document}
\maketitle

\begin{abstract}                          
The object of this paper is the identification of Hammerstein systems, which
are dynamic systems consisting of a static nonlinearity and a linear
time-invariant dynamic system in cascade.
We assume that the nonlinear function can be
described as a linear combination of $p$ basis functions. We model the system
dynamics by means of an $np$-dimensional vector. This vector, usually referred
to as \emph{overparameterized vector}, contains all the combinations between
the nonlinearity coefficients and the first $n$ samples of the impulse response
of the linear block. The estimation of the overparameterized vector is
performed with a new regularized
kernel-based approach. To this end, we introduce a novel kernel
tailored for overparameterized models, which yields estimates that
can be uniquely decomposed as the combination of an impulse response and $p$
coefficients of the static nonlinearity. 
As part of the work, we establish a clear connection between the proposed
identification scheme and our recently developed nonparametric method based on
the stable spline kernel.
\end{abstract}


\section{Introduction}
 A nonlinear system is usually called an
Hammerstein system when it is composed of two blocks in cascade, the first being
a static nonlinearity and the second a linear time-invariant (LTI) dynamic
system~\cite{ljung1999system}.

There are several areas in science and engineering where Hammerstein systems
find applications, see
e.g.~\cite{hunter1986identification},~\cite{westwick2001separable},~\cite{bai2009identification}. For this reason, in
recent years Hammerstein system identification has become a popular and
rather active research topic~\cite{schoukens2011parametric},~\cite{han2012hammerstein}.

Several approaches have been proposed for Hammerstein system identification. For
instance, in~\cite{greblicki1986identification} a kernel-based regression method
is described,~\cite{greblicki2002stochastic} proposes an identification
approach based on a stochastic approximation, while~\cite{goethals2005subspace}
focuses on subspace methods.
In~\cite{westwick2001separable},~\cite{bai2004convergence} and~\cite{rangan1995new} iterative methods based on
least-squares are studied.

An interesting approach was proposed by Er-Wei Bai in~\cite{bai1998optimal}.
Here, the static nonlinearity is modeled as the linear combination of $p$ basis
functions, while the LTI system is assumed to be a finite impulse response (FIR)
with $n$ coefficients. The Hammerstein system is then modeled as a linear
regression, where the regressor vector is $np$-dimensional. Since it contains
all the combinations of the nonlinearity coefficients and the FIR coefficients,
this vector is usually called \emph{overparameterized vector}. Its estimate is
obtained via least-squares and then it is decomposed in order to obtain the
nonlinearity coefficients and the impulse response. Albeit proven to be
asymptotically consistent, the whole procedure suffers of two main drawbacks.
First, since it relies on a least-squares estimation of a possibly very
high-dimensional vector, the final estimates may suffer from high
variance~\cite{ljung1999system}. Second, the procedure does not guarantee that the estimated
$np$-dimensional vector can be exactly decomposed to obtain the nonlinearity
coefficients and the FIR system, and thus approximations are required.

In this paper, we propose a regularization technique to curb the variance of
the estimates of the overparameterized vector. Similarly to~\cite{bai1998optimal}, we model the
Hammerstein system dynamics using the aforementioned overparameterized vector, then we solve
the regression problem relying on a kernel-based approach. To this end, we introduce a
novel kernel, called the \emph{Kronecker overparameterized} (KOP) kernel, which is
the composition of a rank-one positive semi-definite matrix and the so-called
\emph{first-order stable spline kernel}
(see~\cite{pillonetto2010new},~\cite{pillonetto2011prediction},~\cite{bottegal2013regularized}, and~\cite{pillonetto2014kernel} for details). The
structure of this kernel depends on a few parameters (also called
\emph{hyperparameters} in this context), which we need to estimate from data.
This task is addressed by an \emph{empirical Bayes}
approach~\cite{maritz1989empirical}, that is to say by maximizing the marginal likelihood (ML) of the
output. Once the kernel parameters are fixed, the overparameterized vector is
estimated via regularized least squares~\cite{pillonetto2014kernel}. Equivalently, we can think
of the overparameterized vector as a Gaussian random vector with zero-mean and
covariance matrix given by the KOP kernel. With this interpretation, the estimate corresponds to the
minimum mean square error estimate in the Bayesian sense~\cite{wahba1990spline}.

A contribution of this paper is to reveal some interesting properties of the
estimated overparameterized vector provided by the proposed method. We prove
that, as opposed to~\cite{bai1998optimal}, this estimate can be decomposed exactly in order to obtain the nonlinearity
coefficients and the LTI system impulse response, with no loss of information due to approximations. The concept of \emph{exact
decomposition} will be made clear throughout the paper. We also demonstrate strong
connections with our recently proposed method~\cite{risuleo2015kernel},
effectively proving that, although the two approaches are inherently different, the estimates obtained with the two methods are
equivalent.
Finally, we show, through several numerical experiments, that the
proposed method outperforms both the algorithm proposed in~\cite{bai1998optimal} and the standard \matlab\ system identification toolbox function for Hammerstein system
identification.

The paper is organized as follows. In the next section, we formulate the
Hammerstein system identification problem. In Section~\ref{sec:overparameterized}, we describe the modeling approach based on
overparameterized vectors. In Section~\ref{sec:new_approach}, we introduce the
proposed identification scheme, and we give some theoretical background in
Section~\ref{sec:properties}. Numerical experiments are illustrated in Section~\ref{sec:experiments}. Some conclusions end the paper.

\section{Problem formulation}
We consider a stable single input single output discrete-time system described by the
following time-domain relations (see Figure~\ref{fig:hammerstein})
\begin{equation} \label{eq:hammerstein_basic}
\begin{array}{lcl}
	w_t &=& f(u_t) \\
  y_t &=& \sum_{k=1}^{\infty}g_k w_{t-k} + e_t \,.
\end{array}
\end{equation}
In the above equation, $f(\cdot)$ represents a (static) nonlinear function
transforming the measurable input $u_t$ into the unavailable signal $w_t$, which
in turn feeds a strictly causal stable LTI system, described by the impulse
response $g_t$. The output measurements of the system $y_t$ are corrupted by
white Gaussian noise, denoted by $e_t$, which has unknown variance $\sigma^2$.
Following a standard approach in Hammerstein system identification (see e.g.~\cite{bai1998optimal}), we assume that $f(\cdot)$ can be modeled as a
combination of $p$ known basis functions ${\{\phi_i\}}_{i=1}^p$, namely
\begin{equation}
  w_t = f(u_t) = \sum_{i=1}^p c_i \phi_i(u_t) \,,
\end{equation}
where the coefficients $c_i$ are unknown.
\begin{figure}[H]
  \centering
  \includegraphics[width=0.9\columnwidth]{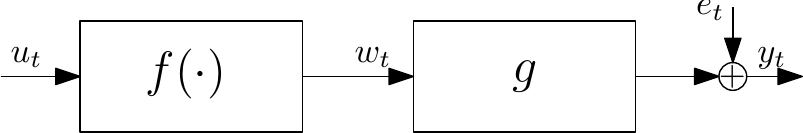}
  \caption{Block scheme of the Hammerstein system.}\label{fig:hammerstein}
\end{figure}
We assume that $N$ input-output samples are collected, and denote them by
${\{u_t\}}_{t=0}^{N-1}$, ${\{y_t\}}_{t=1}^{N}$.
For notational convenience, we also assume null initial conditions.  Then, the
system identification problem we discuss in this paper is the problem of
estimating $n$ samples of the impulse response, say ${\{ g_t\}}_{t=1}^n$ (where
$n$ is large enough to capture the system dynamics), as well as the $p$
coefficients ${\{c_i\}}_{i=1}^p$  characterizing the static nonlinearity $f(\cdot)$.

\subsection{Non-uniqueness of the identified system}\label{sec:identifiability}
It is well-known (see e.g.~\cite{bai2004convergence}) that the two components of a Hammerstein system can be determined up to a scaling
factor. In fact, for any $\alpha \in \mathbb R$, every pair $(\alpha g_t,\,\frac{1}{\alpha}f(\cdot))$, describes the input-output relation equally well.
As suggested in~\cite{bai2004convergence}, we will circumvent this non-uniqueness issue by introducing the following assumption:
\begin{assumption}\label{ass:uniqueness}
The impulse response has unitary $\ell_2$ gain, i.e.
$\norm{g}_2=1$, and the sign of its first non-zero element is positive.
\end{assumption}
\subsection{Notation and preliminaries}
Given a sequence of scalars ${\{a_t\}}_{t=1}^m$, we denote by $a$ its vector representation, i.e.
$$
a =  \begin{bmatrix} a_1 \\ \vdots \\ a_m \end{bmatrix} \quad \in \mathbb R^m \,.
$$

We reserve the symbol $\otimes$ to indicate the Kronecker product of two matrices (or vectors).  We will make use of the bilinear property
$$
(A \otimes B)(C \otimes D) = AC \otimes BD \,,
$$
where $A$, $B$, $C$, and $D$ have proper dimensions. Denoting by $\vect (A)$ the
columnwise vectorization of a matrix $A$, we recall that, for any two vectors
$a$ and $b$, $\vect (ab^T) = b \otimes a$. Given a vector $a \in
\mathbb{R}^{np}$, we introduce its $n \times p$ reshape as
$$
\mathcal R_{n,p}(a) := \matrix{a_1 & \cdots & a_{n(p-1)+1} \\ \vdots & &  \vdots \\ a_n & \cdots & a_{np} } \in \mathbb R^{n \times p} \,.
$$

Given $a \in  \mathbb R^m$, The symbol $\T_n(a)$ denotes the $m \times n$
Toeplitz matrix whose entries are elements of $a$, namely
\begin{equation} \label{eq:toeplitz}
 \T_n(a)\! =\! \begin{bmatrix} a_1 & 0  & & \cdots & 0 \\ a_2 & a_1 & 0 & \cdots
 & 0 \\ \vdots & \vdots &  &\ddots & \vdots \\ a_{m-1} & a_{m-2}  &  \cdots &
 a_{m-n}& 0 \\ a_{m} & a_{m-1}  &  \cdots &\cdots &  a_{m-n+1}  \end{bmatrix} \!
 \in \, \mathbb{R}^{m \times n} \! .
\end{equation}

Let
\begin{equation} \label{eq:shift_matrix}
S = \matrix{0 &\cdots& 0& 0 \\ & & & 0 \\ & I_{m-1} & & \vdots \\ & & & 0} \in \mathbb{R}^{m \times p}\,
\end{equation}
and
\begin{equation}
    P = \matrix{I & S & S^2 & \cdots & S^{n-1}}\,.
    \label{eq:defineP}
\end{equation}
We have the following result, which will be used throughout the paper.
\begin{lemma}\label{lemma:defineP}
  Let $a \in \mathbb R^{m}$  and $\T_n(a)$ be as in~\eqref{eq:toeplitz}. Then
  \begin{equation}
    P\big(I \kron a\big) = \T_n(a)\,,
    \label{eq:lemmaP}
    \end{equation}
\end{lemma}
\begin{proof}
 Note that
  \begin{align}
    \T_n(a) &= \matrix{a & Sa & S^2a& \cdots & S^{n-1}a}\nonumber\\
  &= \matrix{I & S & S^2 & \cdots& S^{n-1}}\matrix{a & & \\ & \ddots & \\ & & a} \nonumber\\
    &=P\big( I \kron a\big) \,,
  \end{align}
  which proves the statement.
\end{proof}
Based on the equality stated by Lemma~\ref{lemma:defineP}, we extend the
Toeplitz notation to matrices, that is, given $A \in \mathbb R^{m \times p}$ we
write
\begin{equation}
\T_n(A) = P\big(I_n \kron A\big) \quad \in \mathbb{R}^{m \times np}\,.
    \label{eq:lemmaF}
  \end{equation}

\section{Identification via overparameterized models}\label{sec:overparameterized}
In this paper, we deal with overparameterized approaches to Hammerstein system
identification. To this end, we construct the matrices
\begin{equation} \label{eq:effona}
 F \triangleq \matrix{\phi_1(u_0)  & \cdots & \phi_p(u_0)\\
                 \vdots& \vdots & \vdots\\
                 \phi_1(u_{N-1})  & \cdots & \phi_p(u_{N-1})\\} \,,\,
\end{equation}
and
\begin{equation} \label{eq:phiona}
\Phi \triangleq  \T_n(F) \quad \in \mathbb{R}^{N \times np}\,.
\end{equation}
Then, we can express the Hammerstein system dynamics problem by means of the
linear regression model (see also~\cite{bai1998optimal})
\begin{equation} \label{eq:bai_regression}
y = \Phi \vartheta + e \,,
\end{equation}
where
\begin{equation} \label{eq:theta}
\vartheta = g \otimes c \, \in \mathbb{R}^{np} \,.
\end{equation}
This vector contains the $n+p$ unknown parameters of the Hammerstein model. Thus, it
constitutes an overparameterization with respect to the original parameters
$c$ and $g$.
A desirable property of any estimate of $\vartheta$ is that it should be
expressible as~\eqref{eq:theta},
namely as a Kronecker product of two vectors. We formalize
this concept in the following definition.
\begin{definition}
Let $\vartheta \in \mathbb{R}^{np}$. We say that $\vartheta$ is a Kronecker
overparameterized (KOP) vector if there exist $g \in \mathbb{R}^{n}$ and $c\in
\mathbb{R}^{p}$ such that~\eqref{eq:theta} holds.
\end{definition}
The following lemma gives a property of KOP vectors.
\begin{lemma}\label{th:hop_vec}
Let $\vartheta \in \mathbb{R}^{np}$ be a KOP vector. Then $\mathcal R_{n,p}(\vartheta) = cg^T$ and thus $\mathcal R_{n,p}(\vartheta)$ is a rank-one matrix.
\end{lemma}
\begin{proof}
Follows from the identity $\vect (cg^T) = g \otimes c$.
\end{proof}

Under Assumption~\ref{ass:uniqueness}, an $np$-dimensional KOP vector
$\vartheta$ can be uniquely decomposed into the $n$- and $p$- dimensional
vectors $g$ and $c$.
\begin{prop}\label{prop:decompose}
  Let $\vartheta\in \mathbb R^{np}$ be a KOP vector. Let $\tilde g$ be the $i$th
  row and $\tilde c$  and the $j$th column of $\mathcal R_{n,p}(\vartheta)$, define
  \begin{equation}\label{eq:reshape}
    g = \frac{\tilde g }{\|\tilde g\|}\,\sign(\tilde g_1), \qquad c =
  \frac{\tilde c}{\tilde g_j}
    \,\|\tilde g\|\,\sign(\tilde g_1).
  \end{equation}
  Then $\vartheta = g\otimes c$, $\|g\|=1$ and $g_1 >0$.
\end{prop}
\begin{proof}
  From~\eqref{eq:reshape} we have that $c_i g_j = \tilde c_i$, so $c_i g_j$ is the $i,j$th
  element of ${\mathcal R_{n,p}(\vartheta)}^T$ hence ${\mathcal
  R_{n,p}(\vartheta)}^T=cg^T$. In addition
  \begin{equation}
    \|g\| = \left\|\frac{\tilde g }{\|\tilde g\|}\,\sign(\tilde g_1)\right\|
    =\frac{\|\tilde g \|}{\|\tilde g\|}\,|\sign(\tilde g_1)| = 1\,,
  \end{equation}
  and
  \begin{equation}
    g_1 = \frac{\tilde g_1 }{\|\tilde g\|}\,\sign(\tilde g_1)
    =\frac{|\tilde g_1 |}{\|\tilde g\|}>0\,,
  \end{equation}
  which completes the proof.
\end{proof}

\subsection{A review of an overparameterized method for Hammerstein system identification}\label{sec:er_wei_method}
In this section we review the identification procedure proposed in~\cite{bai1998optimal}, which constitutes the starting point of our regularized
kernel-based method. Given the model~\eqref{eq:bai_regression}, consistent
estimates of $c$ and $g$ can be obtained with the following steps (see~\cite{bai1998optimal} for details about consistency). First, we compute the
least-squares estimate
\begin{equation} \label{eq:ls_theta}
  \hat \vartheta^{\textrm{LS}} = {(\Phi^T\Phi)}^{-1}\Phi^T y \,.
\end{equation}
Then, since we know that $\vartheta$ is a KOP vector, that is, the reshaping of
$\vartheta$ into an $n \times p$ matrix must be rank-one (Lemma~\ref{th:hop_vec}), we  approximate $\hat \vartheta^{\textrm{LS}}$ to a KOP
vector by approximating $\mathcal R_{n,p}(\hat \vartheta^{\textrm{LS}})$ to a
rank-one matrix. This can be done by solving the problem
\begin{equation} \label{eq:rank1}
\begin{array}{cl}
\textrm{minimize}\ & \|X - \mathcal R_{n,p}(\hat \vartheta^{\textrm{LS}})\|_F \\
\textrm{s.t.} & \rank \,X = 1 \,,
\end{array}
\end{equation}
where $\|\cdot\|_F$ denotes the Frobenius norm.
Expressing $\mathcal R_{n,p}(\hat \vartheta^{\textrm{LS}})$ by means of its
singular value decomposition, i.e.
\begin{align} \label{eq:SVD_ls}
\mathcal R_{n,p}(\hat \vartheta^{\textrm{LS}}) & = USV^T \\
  & = [u^1 \,\, \ldots \,\, u^p] \diag\{s_1,\,\ldots,\,s_p\} {[v^1 \,\, \cdots
\,\, v^p]}^T \nonumber,
\end{align}
we find that the solution of~\eqref{eq:rank1} is $X = u^1 s_1 v^{1T}$. Then
$\hat g = v^1 \sign (v_1^1)$ (since we have assumed $\|g\|_2=1$) and $\hat c =
s_1 u^1 \sign (v_1^1)$.

Note that, since in general $\vartheta^{\textrm{LS}}$ in not a KOP vector,
generally $s_2,\dots,s_p>0$ and the truncation required by the approximation~\eqref{eq:rank1} introduces a bias in the estimates $\hat g$ and $\hat c$ that
degrades performance (see~\cite{hjalmarsson2004direct}). Another drawback of this method
is that it requires the least-squares estimate of the possibly high-dimensional
vector $\hat \vartheta^{\textrm{LS}}$. Hence, despite its consistency property,
the procedure can suffer from high variance in the estimates when $N$ is not
(very) large.

\section{A regularized overparameterization method for Hammerstein system identification}\label{sec:new_approach}
In the previous section, we have seen that the estimator proposed in~\cite{bai1998optimal} suffers from high variance and from a bias that degrades
performance. To control the variance of the estimate, we can use regularization
(for a full treatment, see~\cite{bishop2006pattern},~\cite{hastie2005elements}); this
means we have to select some properties we want to enforce on the estimated
vector.  As we have pointed out in the previous section, a vector
$\hat \vartheta$ is a good candidate estimate of the unknown vector if it satisfies
the following properties:
\begin{enumerate}
\item $\hat \vartheta$ is a KOP vector, so that it can be decomposed as in~\eqref{eq:theta};
\item The mean square error of $\hat \vartheta$ is low, so that the estimated
  values $\hat g$ and $\hat c$ are close to the true values.
\end{enumerate}
A natural approach to incorporate (at least) the second property is based on
regularization or, equivalently, on the Gaussian regression
framework~\cite{rasmussen2006gaussian}. Thus, we model $\vartheta$ as a Gaussian random vector, namely
\begin{equation} \label{eq:prior_theta}
\vartheta \sim \mathcal{N}(0,\,H(\rho)) \,.
\end{equation}
where the covariance matrix (also called a \emph{kernel}) $H(\rho)$ is
parameterized by the vector $\rho$. The structure of $H(\rho)$ determines the
properties of the realizations from~\eqref{eq:prior_theta} and, consequently, of
the estimates of $\vartheta$. In the next subsection, we focus on designing a
kernel suitable for Hammerstein system identification that incorporates also the
first property.
\subsection{The KOP kernel}
We first recall the kernel-based identification approach for LTI systems
proposed in~\cite{pillonetto2010regularized},~\cite{pillonetto2011prediction}, and we model also $g$ as a realization
of a zero-mean $n$-dimensional Gaussian process. Then we have
\begin{equation} \label{eq:prior_g}
    g \sim \N(0, K_{\beta}) \,,
\end{equation}
where the kernel $K_{\beta}$ corresponds to the so-called \emph{first-order
stable spline kernel} (or \emph{TC kernel} in~\cite{chen2012estimation}). It is defined as
\begin{equation}\label{eq:ssk1}
  {\{K_\beta\}}_{i,j} \triangleq \beta^{ \max(i,j)} \,,
\end{equation}
where the hyperparameter $\beta$ is a scalar in the interval $[0,\,1)$. 
  The
choice of this kernel is motivated by the fact that it promotes BIBO stable and
smooth realizations. The decay velocity of these realizations is regulated by
$\beta$.
Typical formulations of the stable spline kernel (see e.g.~\cite{pillonetto2014kernel}) include a scaling factor multiplying the kernel, in
order to capture the amplitude of the unknown impulse response. Here such an
hyperparameter is redundant, as we are working under Assumption~\ref{ass:uniqueness}.

To reconcile~\eqref{eq:prior_g} with~\eqref{eq:prior_theta}, we need to ensure
that the transformation $g\kron c$ is a Gaussian vector, when $g$ is Gaussian.
This is possible if $c$ is a (deterministic) $p$-dimensional vector. In this
case $\vartheta$ is an $np$-dimensional Gaussian random vector with covariance
matrix
\begin{equation}
  H(\rho) = \E[\vartheta\vartheta^T] = \E[(g \otimes c){(g \otimes c)}^T] = K_\beta \otimes cc^T \,,
\end{equation}
which is parameterized by the vector $\rho = {[\beta\,\, c^T]}^T$. In this way, we
have defined a new kernel for system identification based on overparameterized vector regression. We formalize this
in the following definition.
\begin{definition}
We the define the Kronecker overparameterized (KOP) kernel as
\begin{equation} \label{eq:hop_kernel}
  H(\rho) \triangleq K_\beta \otimes cc^T ,\, \quad \rho = {[\beta\,\, c^T]}^T \,,
\end{equation}
where $K_\beta$ is as in~\eqref{eq:ssk1}.
\end{definition}
Note that $H(\rho)$ is rank-deficient, its rank being equal to $n$.
Rank-deficient kernels for system identification have also been studied in~\cite{chen2013rank}.
\subsection{Estimation of the overparameterized vector
\texorpdfstring{$\vartheta$}{}}
We now derive the estimation procedure for the vector $\vartheta$. Recalling
that the noise distribution is Gaussian and given the Gaussian description of
$\vartheta$~\eqref{eq:prior_theta}, the joint distribution of $y$ and
$\vartheta$ is Gaussian. Hence, we can write
\begin{equation} \label{eq:joint_Gaussian}
p\left(\begin{bmatrix} y \\ \vartheta \end{bmatrix};\rho,\,\sigma^2 \right) \sim \mathcal N \left( \begin{bmatrix} 0\\0 \end{bmatrix} , \begin{bmatrix} \Sigma_y & \Sigma_{y\vartheta} \\ \Sigma_{\vartheta y} &  H(\rho) \end{bmatrix} \right)\,,
\end{equation}
where $\Sigma_{y\vartheta} = \Sigma_{\vartheta y}^T =   \Phi H(\rho)$ and
\begin{equation} \label{eq:cov_y}
\Sigma_y =  \Phi H(\rho) \Phi^T + \sigma^2I \,.
\end{equation}
In~\eqref{eq:joint_Gaussian} we have highlighted the dependence of the joint
distribution on the vector $\rho$ and the noise variance $\sigma^2$. Assume
these quantities are given; then the minimum mean square error estimate of
$\vartheta$ can be computed as (see e.g.~\cite{anderson2012optimal})
\begin{align} \label{eq:theta_estim}
\hat \vartheta &= \E[\vartheta|y;\,\rho,\,\sigma^2] \\
	& = H(\rho) \Phi^T \Sigma_y^{-1}y \nonumber \,.
\end{align}
To be able to compute~\eqref{eq:joint_Gaussian} we first need to determine
$\rho$ and $\sigma^2$. This can be done by maximizing the ML of the output data
(see e.g.~\cite{pillonetto2014tuning}). Then we have
\begin{align} \label{eq:marg_lik}
\hat \rho,\,\hat \sigma^2 & = \arg \max p(y;\,\rho,\,\sigma^2) \nonumber\\
& = \arg \min_{\rho,\sigma^2} \Big[\log \det \Sigma_y + y^T \Sigma_y^{-1} y
\Big]\,.
\end{align}
The resulting estimation procedure for $\vartheta$ can be summarized by the
following two steps.
\begin{enumerate}
\item Solve~\eqref{eq:marg_lik} to obtain $\hat \rho,\,\hat \sigma^2$.
\item Compute~\eqref{eq:theta_estim} using the estimated parameters.
\end{enumerate}
Having obtained $\hat\vartheta$, it remains to establish how to decompose it in
order to obtain the estimates $\hat g$ and $\hat c$. In the next section we
shall see that, using the proposed approach, such an operation becomes natural.

\section{Properties of the estimated overparameterized vector}\label{sec:properties}
In this section, we analyze some properties of the regularized
overparameterization estimate of $\vartheta$. In particular, we show that the
estimates produced by~\eqref{eq:theta_estim} are KOP vectors. Then, we show that this procedure leads
to exactly the same estimator as the one we proposed in~\cite{risuleo2015kernel}; where the coefficients of the nonlinearity were
considered as model parameters and not included among the kernel hyperparameters.

To further specify the equivalence, we first briefly review the Hammerstein
system identification approach proposed in~\cite{risuleo2015kernel} which
is based on a different Gaussian process assumption.

\subsection{A review of the method proposed in \texorpdfstring{~\cite{risuleo2015kernel}}{[25]}}
 Let $W = \T_n(w) = \T_n(Fc)$. Then we can model the measurements with the
 linear relation
\begin{equation} \label{eq:hammerstein_g}
y = Wg+e \,.
\end{equation}
Modeling $g$ as a Gaussian random vector with covariance given by the stable
spline kernel~\eqref{eq:ssk1}, we notice that a joint Gaussian description holds
between $y$ and $g$. Hence we can write
\begin{equation} \label{eq:joint_Gaussian_g}
p\left(\begin{bmatrix} y \\ g \end{bmatrix};c,\,\beta,\,\sigma^2 \right) \sim \mathcal N \left( \begin{bmatrix} 0\\0 \end{bmatrix} , \begin{bmatrix} \Sigma_{y,2} & \Sigma_{yg} \\ \Sigma_{g y} &  H(\rho) \end{bmatrix} \right)\,,
\end{equation}
where $\Sigma_{yg} = W K_\beta$ and $\Sigma_{y,2} = W K_\beta W^T + \sigma^2I$.
Note that~\eqref{eq:joint_Gaussian_g} depends on the parameters $c$, $\beta$ and
$\sigma^2$. These parameters are estimated via ML maximization, that is by
solving
\begin{equation} \label{eq:marg_lik2}
  \hat c,\,\hat \beta,\,\hat \sigma^2 = \arg \min_{c,\beta,\sigma^2}\Big[\log
  \det \Sigma_{y,2} + y^T \Sigma_{y,2}^{-1} y\Big] \,.
\end{equation}
The minimum mean square estimate of $g$ is then computed as
\begin{equation} \label{eq:g_method_sysid}
\hat g  = \E[g|y,\,\hat c,\,\hat \beta,\,\hat \sigma^2] = K_\beta  W^T \Sigma_{y,2}^{-1} y \,.
\end{equation}
In the next section, we point out the strong connection
between~\eqref{eq:g_method_sysid} and the estimate~\eqref{eq:theta_estim}, produced by
the KOP kernel-based regression approach.

\subsection{The \texorpdfstring{estimate~\eqref{eq:theta_estim}}{estimate (26)} is a KOP vector}
In this section we prove that, when the KOP kernel-based method is used to
estimate~\eqref{eq:theta_estim}, the resulting estimates can be decomposed as
Kronecker products of lower-dimensional vectors and thus they are KOP vectors.
Before arriving at this result we show the equivalence between the output
measurement models~\eqref{eq:bai_regression} and~\eqref{eq:hammerstein_g}.
\begin{lemma}\label{lemma:kernels}
Let $W = \T_n(Fc)$ and $\Phi$ as in~\eqref{eq:phiona}. Then
    \begin{equation}
      \Phi H(\rho) \Phi^T = W K_\beta W \,,
      \label{eq:lemma_kernels}
    \end{equation}
where $H(\rho)$ and $K_\beta$ are the KOP and the stable spline kernels.
\end{lemma}
\begin{proof}
Recalling the bilinear property of the Kronecker product and Lemma~\ref{lemma:defineP}, we see that
  \begin{align*}
    \Phi H(\rho) \Phi^T \!\!\!& = P\big[ I \kron F]H(\rho) \big[ I \kron
    F^T]P^T \nonumber\\
    &= P\big[ I \kron F \big][K_\beta \kron c c^T] \big[ I \kron
    F^T\big]P^T \nonumber\\
    &= P\big[ I \kron F\big][I \kron c ][K_\beta \kron 1] [I \kron c^T] \big[ I\! \kron\!  F^T\!\big]\!P^T \nonumber\\
    &= P\big[ I \kron Fc\big][K_\beta \kron 1]  \big[ I \kron
    c^T F^T\big]P^T \nonumber\\
    &= \T_n(Fc)K_\beta {\T_n(Fc)}^T=WK_\beta W^T \,,\nonumber
  \end{align*}
  which proves the result.
\end{proof}
\begin{theorem}\label{thm:marginals_equal}
Consider the output measurement models~\eqref{eq:bai_regression} and~\eqref{eq:hammerstein_g}. Then:
\begin{enumerate}
\item The marginal likelihoods of $y$ obtained from the two models are equivalent;
\item The parameter estimates obtained from~\eqref{eq:marg_lik} and~\eqref{eq:marg_lik2} are the same.
\end{enumerate}
\end{theorem}
\begin{proof}
Let
\begin{align}
p_1(y;\rho,\,\sigma^2) & = \int p\left(y,\,\vartheta;\,\rho,\,\sigma^2 \right) d\vartheta  \\
p_2(y;c,\,\beta,\,\sigma^2) & = \int p\left( y ,\,g;\,c,\,\beta,\,\sigma^2 \right) dg
\end{align}
be the marginal likelihoods of the two models. We find that
  \begin{align*}\label{eq:proof_loglik}
    p_1(y;\rho,\,\sigma^2) & = \N(0,\Sigma_{y}),\\
    p_2(y;c,\,\beta,\,\sigma^2) & = \N(0,\Sigma_{y,2})\,.
  \end{align*}
Using Lemma~\ref{lemma:kernels}, we have that $\Sigma_{y}=\Sigma_{y,2}$, hence
$p_1$ and $p_2$ are equivalent. The same promptly holds for their ML maximizers
$\hat \rho = {[\hat\beta \,\,\hat c^T]}^T$ and $\hat \sigma^2$.
\end{proof}
We are now in the position to prove that the estimate $\hat \vartheta$ is a KOP
vector
\begin{theorem}\label{th:hop_vector_estimate}
Assume that $\rho$ and $\sigma^2$ are estimated using the ML
approach~\eqref{eq:marg_lik} (or, equivalently,~\eqref{eq:marg_lik2}). Then, the minimum
variance estimate of $\vartheta$ in~\eqref{eq:theta_estim} is such that
  \begin{equation} \label{eq:estim_hop}
   \hat \vartheta = \hat g \kron \hat c \,,
  \end{equation}
where $\hat g$ is the minimum variance estimate of $g$ in~\eqref{eq:g_method_sysid} and $\hat c$ is the ML estimate of $c$.
\end{theorem}
\begin{proof}
Using~\eqref{eq:theta_estim} and recalling the bilinear property of the
Kronecker product and Lemma~\ref{lemma:defineP}, we have
\begin{align}\label{eq:proof_thm1}
    \hat \vartheta & = \Sigma_{\vartheta y} \Sigma_{y}^{-1} y = H(\hat \rho)
    \Phi^T\Sigma_{y}^{-1}y \nonumber \\
    &= [K_{\hat \beta} \kron \hat c\hat c^T][I\kron F^T]P^T\Sigma_{y}^{-1}y \nonumber\\
    &= [K_{\hat \beta} \kron \hat c][I\kron \hat c^T F^T]P^T\Sigma_{y}^{-1}y \\
    &= [K_{\hat \beta} \kron \hat c]W^T\Sigma_{y}^{-1}y \nonumber\\
    &= [K_{\hat \beta} \kron \hat c][W^T\Sigma_{y}^{-1}y\kron 1] \nonumber \\
    &= [K_{\hat \beta} W^T\Sigma_{y}^{-1}y\kron \hat c] \nonumber\,.
  \end{align}
From Theorem~\ref{thm:marginals_equal} we know that $\Sigma_{y} = \Sigma_{y,2}$;
thus, recalling~\eqref{eq:g_method_sysid} we have
\begin{equation}
K_{\hat \beta} W^T\Sigma_{y}^{-1}y = \hat g \,,
\end{equation}
so that~\eqref{eq:estim_hop} is obtained.
\end{proof}
\begin{corollary}\label{cor:rank1}
The estimate $\hat \vartheta$ given in~\eqref{eq:theta_estim} is a KOP vector and $\mathcal R_{n,p} (\hat \vartheta)$ is rank-one.
\end{corollary}
\begin{proof}
Since from~\eqref{eq:estim_hop} we have $\hat \vartheta = \hat g \kron \hat c$,
$\hat \vartheta$ is a KOP vector. The second part of the statement follows
directly from Lemma~\ref{th:hop_vec}.
\end{proof}
We can make an interesting observation, that further links the KOP estimate to
our previous kernel based estimator:
\begin{corollary}
  The estimates of the nonlinearity coefficients $\hat c$, found
  maximizing~\eqref{eq:marg_lik} and those resulting from decomposing
  $\hat \vartheta$ as in~\eqref{eq:estim_hop} are the same.
\end{corollary}
\begin{proof}
  Follows directly from~\eqref{eq:proof_thm1} and Theorem~\ref{thm:marginals_equal}.
\end{proof}

We have established that the estimate $\hat \vartheta$ produced using the
procedure detailed in Section~\ref{sec:new_approach} is a KOP vector. So, the
estimates of the impulse response $\hat g$ and the nonlinearity coefficients
$\hat c$ can be retrieved using~\eqref{eq:reshape}. The whole procedure is
summarized in  Algorithm 1.

\begin{figure*}[t!]\label{fig:boxplot}
  \begin{center}
    \includegraphics[width=.9\textwidth]{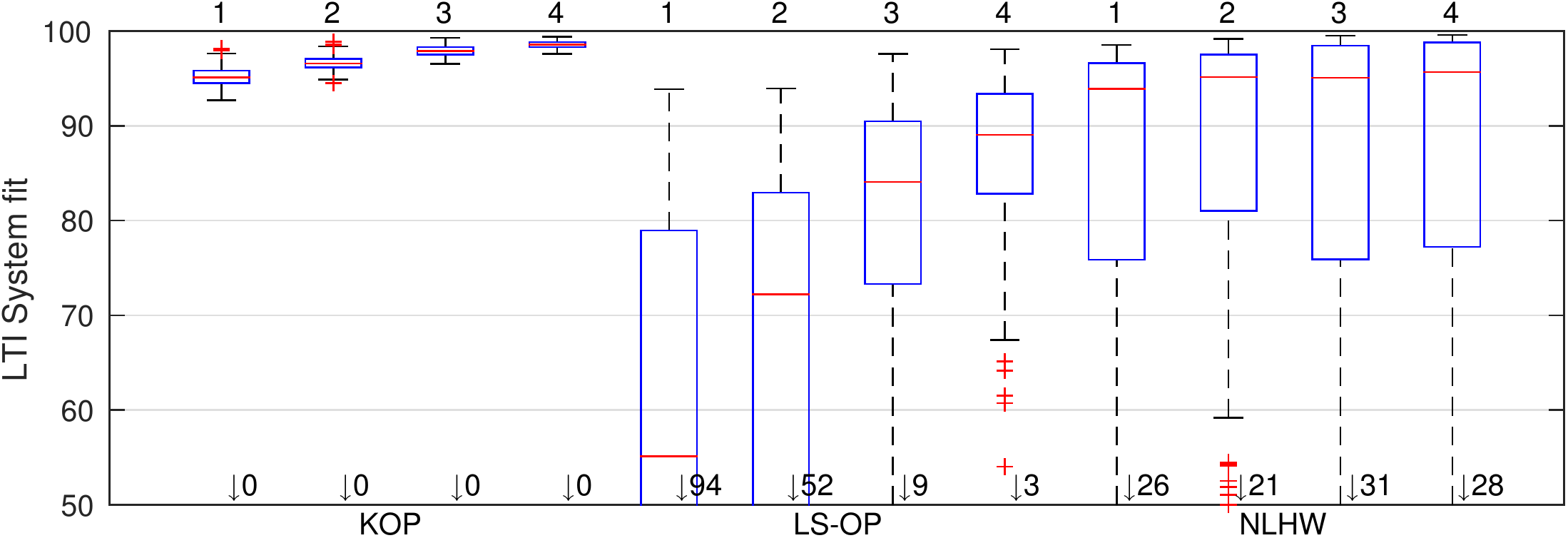}\\
    \includegraphics[width=.9\textwidth]{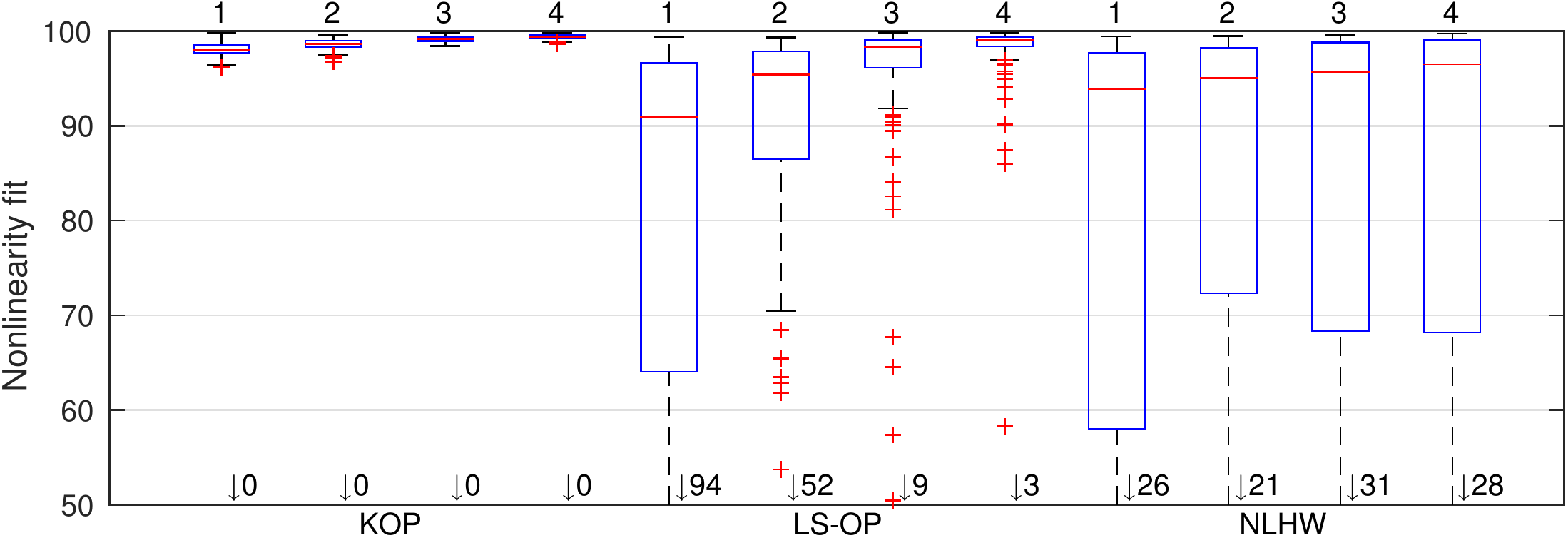}
  \caption{Results of the Monte Carlo experiments for different SNR\@. Top: Fit in
  percent of the linear system impulse response. Bottom: Fit in percent of the
nonlinear transformation.}
    \end{center}
\end{figure*}

\begin{algorithm}[ht!]\label{alg}
\textbf{Algorithm 1}: KOP kernel-based Hammerstein system identification \vspace{0.1cm}\\
Input: ${\{y_t\}}_{t=1}^{N}$, ${\{u_t\}}_{t=0}^{N-1}$ \\
  Output: ${\{\hat{g}_t\}}_{t=1}^n$, ${\{\hat{c}_i\}}_{i=1}^p$
  \begin{enumerate}
  \item Obtain $\hat \rho$, $\hat \sigma^2$ solving~\eqref{eq:marg_lik}
  \item Estimate $\hat \vartheta$ using~\eqref{eq:theta_estim}
  \item Find $\hat g$, normalizing $\hat \vartheta$ by $\hat c$ (Proposition~\ref{prop:decompose}).
  \end{enumerate}
\end{algorithm}
The result of the outlined regularization procedure
applied to the overparameterized vector, with a suitable rank deficient kernel,
yields estimates that are equivalent to the ones provided by the procedure
outlined in~\cite{risuleo2015kernel}.
\section{Numerical Experiments}\label{sec:experiments}
We evaluate the proposed algorithm with numerical simulations of Hammerstein
systems. We set up 4 experiments in order to test different experimental
conditions. The experiments consist of 200 independent Monte Carlo runs each.
At every Monte Carlo run, we generate random data and Hammerstein systems,
according to the following specifics.
\begin{itemize}
\item The linear subsystem model is of output error type:
  \begin{equation}
    y(t) = \frac{B(q)}{A(q)}f(u(t)) + e(t),
  \end{equation}
 generated by picking $4$ poles and $4$ zeros at random. The
  poles and zeros were sampled in conjugate pairs $(a e^{j\omega},
  ae^{-j\omega})$ with $a$ uniform in $[0.5,0.95]$ and $\omega$ uniform in
    $[0,\pi]$.
\item The input nonlinearity is a polynomial of fourth order. It is a linear
  combination of Legendre polynomial basis functions, defined as
\begin{equation}
  P_i(u)= \frac{1}{2^i i!}\frac{\partial^i}{\partial u^i}\left[ {(u^2-1)}^i \right] \,,
\end{equation}
where $i = 0,\,\ldots,\,4$. The coefficients $c$ are chosen uniformly in
$[-1,1]$.
\item The input to the system is Gaussian white noise with unit variance.
\item The experimental data consists in $N=1000$ pairs of input-output samples,
 simulated from zero initial conditions.
\item   The measurement noise $e(t)$ is Gaussian and white. Its variance is a fraction of the noiseless
  output variance, i.e.
\begin{equation}
  \sigma^2 = \frac{\Var{Wg}}{SNR} \,,
\end{equation}
where $SNR$ depends on the experiment.
\end{itemize}
Every experiment is carried out in a different signal to noise ratio (SNR)
condition, see Table I.
\begin{table}[ht]\label{tab:fits}
\begin{center}
\begin{tabular}{ccccc}
  \toprule
Experiment $\#$ & 1 & 2 & 3 & 4\\
\midrule
SNR & 10 & 20 & 50 & 100\\
\bottomrule
\end{tabular}
\caption{SNR considered in the experiments}
\end{center}
\end{table}

We aim at estimating $n=30$ samples of the system impulse
response of the LTI systems (which are such that $\|g\|_2=1$ and with the sign of
the first sample positive) and the $p=5$ coefficients of the nonlinear block.
We test the following estimation methods.
\begin{itemize}
  \item \emph{KOP}: This is the method described in this paper. The ML
    optimization problem is solved using the function \texttt{fminsearch}
    available in \matlab.
    The search was initialized with the elements of $c$ uniformly sampled in
    $[0,1]$, $\beta_0=0.5$, and $\sigma$ equal to the sample variance of the
    residuals of $\hat\vartheta^{\textrm{LS}}$.
  \item \emph{LS-OP}: This estimator implements the least-squares
    overparameterization-based method proposed in~\cite{bai1998optimal} and
    briefly reviewed in Section~\ref{sec:er_wei_method}. Note that, under the
    working experimental conditions, this method has to perform a least-squares
    estimate of a $150$-dimensional vector.
  \item \emph{NLHW}: This is the \matlab\ function \texttt{nlhw} that uses the
    prediction error method to identify the linear block in the system
    (see~\cite{ljung2009developments} for details). To get the best performance from
    this
    method, we  equip it with an oracle that knows the true order of the LTI
    system generating the measurements and knows the order of the polynomial input nonlinearity.
\end{itemize}
Note that all the methods
  have available the same amount of prior information, namely the orders of the
  input polynomial. The knowledge of th order of the linear block is known only to NLHW, which makes use of a parametric description of the linear system. Furthermore, we
  note that, due to the Gaussianity of the noise, the least-squares procedure
  in LS-OP is optimal in the Gauss-Markov sense.

We assess the accuracy of the estimated models using two performance indices.
The first is the fit of the system impulse response, defined as
\begin{equation}
  FIT_{g,i} \triangleq 100\left(  1 - \frac{\norm{g_i - \hat g_i}_2}{\norm{g_i - \bar g_i}_2}\right)\,,
\end{equation}
where $g_i$ is the system generated at the $i$-th run of each experiment,
$\hat g_i$ its estimate and $\bar g_i$ its mean. The second is the fit of the static
nonlinear function, given by
\begin{equation}
  FIT_{f,i} \triangleq 100\left( 1 - \frac{\Vert f_i(u) - \hat f_i(u)\Vert_2}{\Vert f_i(u) -
  \overline{f_i(u)}\Vert_2} \right) \,.
\end{equation}
Figure 1 shows the results of the outcomes of the 4 experiments. The box plots
compare the results of KOP, LS-OP and NLHW for the considered signal to
noise ratios. We can see that, for high SNR, all the estimators perform well,
especially in identifying the nonlinearity coefficients. For lower SNR
, however, the proposed method KOP performs substantially better than the
others. This is mainly because of the regularizing effect of the KOP kernel that
reduces the variance of the estimates. Notice also that the proposed
approach enforces the rank deficiency in the reshaped version of $\hat \vartheta$, so
it circumvents the errors introduced by the rank-one approximation made by
LS-OP\@.
The main drawback of NLHW is that it relies on a high dimensional nonlinear
optimization, as it needs to estimate all the parameters in the model. The
proposed method is instead nonparametric, and does not rely on the knowledge of
the order of the LTI system.

\section{Conclusions}
Regularization is an effective technique to control the variance of least
squares estimates. In this paper we have studied how to improve popular overparameterization methods for
Hammerstein system identification using Gaussian process regression with a
suitable prior. To this end, starting from the stable spline kernel, we
have introduced the KOP kernel, which we believe to be a novel concept in
Hammerstein system identification. Using the KOP kernel, we have designed a
regularized least-squares estimator which provides an estimate of the
overparameterized vector. The impulse response of the LTI system and the
coefficients of the static nonlinearity are then retrieved by suitably
decomposing the estimated vector. In contrast with the original
overparameterization method, this decomposition involves no approximation.
An important contribution is showing that this procedure
estimate is equivalent to our recently proposed kernel-based
method~\cite{risuleo2015kernel}. Using
simulations, we have shown that the proposed method compares very favorably
with the current state-of-the-art algorithms for Hammerstein system
identification.

The introduction of the KOP kernel possibly opens up for new effective system identification methods based on the combination of overparameterized vectors and regularization techniques. In fact, we believe that Hammerstein system identification is not the only problem where KOP kernels could find application.
Another possible extension of the proposed method is the design of new kernels
merging a kernel for the static nonlinearity and the stable spline kernel.
The main issue with this approach is that, at least theoretically, the Gaussian
description of the resulting overparameterization vector would be lost.

\printbibliography\
 \end{document}